\documentclass[a4paper,11pt]{article}

\usepackage{fullpage}
\usepackage{amsmath}
\usepackage{amssymb}
\usepackage[dvips]{graphicx}
\usepackage{natbib}
\usepackage{setspace}
\usepackage{lineno}

\usepackage{amsthm}
\usepackage{amsmath}
\usepackage{relsize}
\usepackage{xcolor,calc}

\newtheorem{Remark}{Remark}

\newcommand{\limprob}{\xrightarrow{P}}
\newcommand{\limdist}{\xrightarrow{D}}
\newcommand{\chords}{\textit{chords}}
\newcommand{\gauss}[1]{\mathcal{N}\left(#1\right)}

\newcommand{\Naive}{Na\"{\i}ve}
\newcommand{\set}[1]{\left\{ #1 \right\}}

\title{\textbf{Modeling and Analyzing Respondent-Driven Sampling as a Counting Process}}
\author{Yakir Berchenko\\University of Cambridge,\\Madingley Road, Cambridge CB3 0ES, United Kingdom\\and\\ Biostatistics Unit,\\Gertner Institute for Epidemiology and Health Policy Research,\\Tel Hashomer 52621, Israel\\(Corresponding author: {\tt byakir@gmail.com})
\and Jonathan Rosenblatt\\Department of Computer Science and Applied Mathematics,\\ 
Weizmann Institute of Science,\\
Rehovot, Israel
\and Simon D.W. Frost\\Department of Veterinary Medicine,\\University of Cambridge,\\Madingley Road, Cambridge CB3 0ES, United Kingdom}

\begin{document}

\maketitle

\begin{abstract}
Respondent-driven sampling (RDS) is an approach to sampling design and analysis which utilizes the networks of social relationships that connect members of the target population, using chain-referral methods to facilitate sampling. 
RDS will typically lead to biased sampling, favoring participants with many acquaintances. 
\Naive\ estimates, such as the sample average, which are uncorrected for the sampling bias, will themselves be biased towards the state of the most highly connected individuals.
To compensate for this bias, current methodology suggests inverse-degree weighting, where the ``degree'' is the number of acquaintances. 
This stems from the fundamental RDS assumption that the probability of sampling an individual is proportional to their degree.
Since this assumption is tenuous at best, we propose to harness the additional information encapsulated in the time of recruitment, into a model-based inference framework for RDS.
This information is typically collected by researchers, and ignored. 
We adapt methods developed for inference in epidemic processes to estimate the population size, degree  counts and frequencies.
While providing valuable information in themselves, these quantities ultimately serve to debias other estimators, such a disease's prevalence.
A fundamental advantage of our approach is that, being model-based, it makes all assumptions of the data-generating process explicit. 
This enables verification of the assumptions, maximum likelihood estimation, extension with covariates, and model selection.
We develop asymptotic theory, proving consistency and asymptotic normality properties. 
We further compare these estimators to the standard inverse-degree weighting through simulation, and using real-world data. 
In both cases we find our estimators to outperform current methods. 
The likelihood problem in the model we present is coordinatewise convex, and thus efficiently solvable. 
We implement these estimators in an accompanying R package, \chords, available on CRAN. 
\end{abstract}

\section{Introduction}
\label{sec:introduction}

Marginalized populations often suffer a disproportionate burden of infectious disease. Yet the hard-to-reach or hidden nature of these populations makes them difficult to sample, limiting our knowledge of the very groups for which monitoring and prevention should be a priority. 
Respondent-driven sampling (RDS) is an approach to sampling design that is increasingly  used to study marginalized or highly stigmatized groups, such as injection drug users, sex workers, and men who have sex with men \citep{Heckathorn1997,Heckathorn2002}. 
RDS overcomes the hidden nature of these populations by utilizing the networks of social relationships that connect their members together, to facilitate sampling by chain-referral methods. 
Seeds are selected --- usually by convenience --- from the target population, and given coupons. 
They use these coupons to recruit others, who themselves become recruiters. 
Recruits are given an incentive, usually money, for taking part in the survey, and also for recruiting others. 
This process continues in recruitment waves until the survey is stopped, usually when a target sample size is reached. 
Estimation methods are then applied which account for the nonuniform sample selection, in an attempt to generate unbiased estimates of the composition of the population.  
Following its introduction  \citep{Heckathorn1997,Heckathorn2002}, RDS has quickly become popular. Major public health agencies, such as the World Health Organization and the Centers for Disease Control and Prevention, now rely on it for HIV surveillance activities, such as prevalence estimation.

As pointed out by \citet{SalganikCom}, RDS is a package consisting of two distinct components, namely a sampling method and a method of statistical inference. 
With some notable exceptions, the sampling method has often been found to be cost-effective for attaining a target sample size, and the popularity of RDS has led to a wealth of new data \citep{Malekinejad2008}. 
However, the assumptions and performance of the second, inferential, component of the RDS package are far more vulnerable to criticism \citep{crawford2015hidden,guntuboyina2012impossibility}.

The most fundamental inferential problem in RDS is biased sampling, such as oversampling participants with many acquaintances. This biases population estimates towards the state of these highly connected participants.
The best way to address biased sampling is to  stratify the sample into different degree classes; this is done by measuring each participant's degree, and estimating the prevalence, $H$, as a weighted average,

\begin{equation} \label{Hstrat}
\widehat{H} = \sum_{k\geq 1} f_{k}\widehat{p_k}
\end{equation}
where $\{f_k\}_{k\geq 1}$ is the degree distribution of the population, and $\widehat{p_k}$ is the estimator of  $p_k$, the prevalence within degree class $k$.
This is not usually possible, however, because the real degree distribution of the population, $f_k$, is not known; moreover, when using only the observed degrees, or order of recruitment, the degree distribution is not identifiable. 
Denoting the degree-dependent sampling probability of an individual by $\pi_k$, it is possible to estimate only the probability of recruitment ($ \pi_k f_k$), but not $f_k$ and $\pi_k$ separately.
Current RDS studies attempt to rectify this problem by modeling recruitment as a \emph{homogenous} random walk, which culminates in the assumption that the sampling probability is proportional to degree, i.e., $ \pi_k \propto k$. 
Under this assumption the current inverse-degree weighting is in fact an inverse-probability estimator, i.e., a Horvitz--Thompson estimator \citep{Horvitz_HT}.

We hypothesize, however, that it is highly unlikely that the recruitment probability is indeed proportional to degree. 
Other means should thus be invoked to restore the identifiability of the degree frequencies. 
Moreover, information on the observed degree distribution, by itself, does not allow us to test this assumption.  
Fortunately, most RDS studies obtain additional valuable information which is usually discarded: the precise timing of recruitment\footnote{Readers interested in size-biased sampling without replacement, where the sampling time is not known and only the \emph{order} of sampling is known, should consult the seminal works of \citet{Gordon1993} and \citet{NairWang1992}.
}. 
This timing information allows us to identify $f_k$ and $\pi_k$. 
We harness it by modeling recruitment as a continuous-time counting process, and utilize the established machinery \citep{Andersen93} which has been applied, for example, to survival analysis in stochastic epidemic models \citep{AB2000} and software reliability \citep{VanPul}.

Our approach, discarding the homogenous random-walk model in favor of a stochastic epidemic model, is a very natural one. The recruitment process is akin to the spread of an epidemic in a population; hence, why not model it as one? This is particularly promising since it involves linking RDS to a larger, more developed corpus of literature \citep{AB2000, Britton98, Rida91, Becker89}.

In Section~\ref{sec:results}, after introducing our new model for RDS, we discuss the related literature of epidemiological modeling and inference, as well as certain related models dealing with inference for software reliability. 
We derive the associated maximum likelihood estimators (MLEs) for our model, and demonstrate their utility empirically, using simulations and real-life data. 
Some large-sample properties, such as consistency and asymptotic normality, follow after the appropriate technical preliminaries.
We end with a discussion emphasizing the benefits of our model-based approach to RDS, and some future research directions, in Section \ref{sec:discussion}.

\section{Results}
\label{sec:results}
We begin by introducing our new statistical model. We then derive the associated MLEs, and discuss their properties.
Our notation here attempts to maintain compatibility with both epidemiological modeling \citep{AB2000} and the theory of inference for continuous-time counting processes \citep{Andersen93, VanPul}. Minor unavoidable clashes are resolved as explained below.

\subsection{Problem Setup} 
\label{model}
Our approach for modeling RDS assumes the following:
\begin{description}
  \item[(M1)] The size of the population, $N$, is not known, although we may assume it is very large.
  \item[(M2)] For each degree, $k$, there are $N_k$ individuals in the population with degree $k$.
  \item[(M3)] Sampling is done without replacement, with $n_{k,t}$ being the (right-continuous) counting process representing the number of people with degree $k$ recruited by time $t$.
  \item[(M4)] Between times $t$ and $t+\Delta t$, an individual with degree $k$ is sampled with probability
\begin{equation}
\label{eq:intensity}
 \lambda_{k,t} =  \frac{\beta_k}{N}I_t(N_k - n_{k,t})\Delta t + o(\Delta t)
\end{equation}
where $I_t$ is the number of people already recruited and actively trying to recruit new individuals, and the constant $\beta_k$ is a degree-dependent recruitment rate.

\end{description}
Using $g_t^-$ to denote\footnote{Hopefully less cumbersome than the alternative common notation, $g_{t-}$.} the value of $g$ just before $t$, a more formal statement of (M4) is:
\begin{description}
  \item[(M4')] The multivariate counting process $n_t:=(n_{1,t},n_{2,t},..n_{d_{max},t})$ has intensity
\begin{equation} 
\label{intense}
	\lambda_t:=\bigg(\frac{\beta_1}{N}I_{t}^-(N_1 - n_{1,t}^-), \frac{\beta_2}{N}I_t^-(N_2 - n_{2,t}^-),\dots, 		\frac{\beta_{d_{max}}}{N}I_t^-(N_{d_{max}} - n_{d_{max},t}^-)\bigg)
\end{equation}
such that, $m_t:=n_t - \int_0^t \lambda_s ds$ is a multivariate martingale (and clearly, $\lambda_t$ is predictable, i.e., nonstochastic given the past).
\end{description}

\paragraph{Relation to other domains.}
The similarity of Eq.(\ref{intense}) to the widespread Susceptible-Infected-Removed (SIR) epidemiological model \citep{AB2000} is quite striking. In the simplest version of the SIR model\footnote{Most of the more elaborate epidemiological models could also be adapted for RDS. For example, it is also possible to consider the case where the probability of a person recruiting new individuals is proportional to his degree. In this case we need to replace $I_t$ in Eq.(\ref{intense}) with $\tilde{I_t}$, which is the number of edges sampled so far; i.e., if $x_t$ is the observation at time point $t$, with $x_t = 0$ if no one was sampled and $x_t = k$ otherwise ($k$ being the degree of the sampled individual), then $\tilde{I_t} = \int (x_t)dt - I_t$. An even more general recruitment mechanism is considered by \citet{Britton98}, addressing contagion and estimation in multitype epidemics.}, the susceptible set, $S$, is depleted at a rate, $dS$, proportional to its size and the size of the infected set, $I$: $dS = -\beta ISdt$. 
Thus, in RDS the ``inviting'' set is analogous to the infectious set in standard epidemiological modeling. 
Similarly, \mbox{$N-n$} is the analog of the susceptible set, $S$. However, previous epidemiology-related works \citep{Rida91,Britton98,AB2000} have usually focused on the transmission parameters ($\beta_k$'s in our model), which are the least relevant to our application. As such they also assume knowledge of the degree distribution, which in our case is not only unknown, but actually one of the main objects of interest.

These epidemiological model features are complemented by certain models dealing with inference in the field of  software reliability \citep{VanPul,JelinskiMor}. In particular, the Jelinski--Moranda model assumes that a computer program has an unknown number of bugs, $N$, which are detected at a rate proportional to the number of remaining (undetected) bugs; i.e., the rate of detection for the $i^\mathrm{th}$ bug is $\lambda_i = \beta (N-(i-1))$. In this case, the motivation and approach for estimating $N$ is more akin to RDS \citep{VanPul}. 
Two key differences remain.
The first is that the Jelinski--Moranda model is a special case of Eq.(\ref{intense}) with $I_t\equiv N$, whereas in RDS $I_t$ is more general, and depends on the number of individuals detected (\S\ref{asy}).
The second is that, unlike the Jelinski--Moranda model, RDS is multivariate. This is further exacerbated by the fact that often the $N_k$'s themselves are nuisance parameters, required for stratification further down the road.

\paragraph{Short summary of the main results.}
For the model introduced above we prove the following main results:
\begin{enumerate}
\item The MLE for $N_k$, denoted $\hat{N}_k$, is given by the (unique) solution to
\begin{equation}
\label{eq:Nknumeric}
 \sum_{i=0}^{n_{k,\tau}-1} \frac{1}{\hat{N}_k - i} = \frac{n_{k,\tau} \int_0^\tau I_t dt }{\hat{N}_k \int_0^\tau I_t dt - \int_0^\tau n_{k,t} I_t dt}
\end{equation}
where $\tau$ is the duration of the RDS survey. 
\item The MLE for $N_k$ is asymptotically consistent and jointly multivariate normal for all $k$'s. 
\item The consequential  MLE for prevalence is asymptotically consistent and normally distributed.
\end{enumerate}

We now demonstrate the empirical performance of our method (\S\ref{sec:empirical}), followed by an asymptotic analysis (\S\ref{sec:asymptotics}).

\subsection{Empirical Results}
\label{sec:empirical}

We have implemented our estimator in the \chords\ R package, which is available from CRAN. 
We use it to examine our method using simulations, as well as a real RDS dataset.

\subsubsection{Numerical implementation}
\label{sec:numerics}

Eq.(\ref{eq:Nknumeric}) demonstrates that the MLEs for each $N_k$ are separable, and the likelihood is coordinatewise convex.  
This allows a simple line search to be used in order to find $\hat{N}_k$.
As a result, the implementation is very fast: most of the effort is dedicated to careful summation of $\int_0^\tau I_t dt$ and $\int_0^\tau n_{k,t} I_t dt$, prior to a coordinatewise line search. 
The efficiency of the estimation allows computationally intensive approaches such as bootstrapping and cross-validation to be applied.

In contrast to other methods \citep{Gile11}, we do not need to make any assumptions about the true degree distribution or population size. If these are available, however, we can easily incorporate them via a Bayesian prior, thereby regularizing the estimation.

\subsubsection{Simulations}
\label{sec:simulation}

We tested our method and compared it via simulations (Table~\ref{tab:MAE} and Figure~\ref{fig:boxes}) to the successive sampling (SS) estimator of \citet{Gile11}, a state-of-the-art multistage inverse-degree weighting estimator. 
We then apply it to a real-world RDS dataset.

We generated data using different generative processes, which differ in the effect of the degree on the probability of being sampled. In terms of Eq.(\ref{eq:intensity}) --- or more precisely, Eq.(\ref{Xjfun}), below --- we set $\beta_j := j^{\theta}$ for $\theta \in \set{0,0.5,1}$.
Setting $\theta=1$ corresponds to the implied model in inverse-degree weighting, $\theta=0$ corresponds to a degree-\emph{independent} model, and $\theta=0.5$ indicates a weak degree dependence.
We also generated data from a misspecified model where the recruitment rate, $\beta_j$, changes randomly after each sample following the law $\gauss{0,j^2}$.
For all $\theta$, we set the population size, $N$, to 2,000, evenly split between two degrees: $N_2=1$,000 and $N_{10}=1$,000.
In each setup, we generated $300$ samples, each of size 1,000.

We then analyzed these RDS samples using our method and the SS estimator.
The SS estimator requires a population size prior ($\tilde{N}$). We tried a correct prior, $\tilde{N}=2$,000, and an incorrect one, $\tilde{N} = 4$,000.

\begin{table*}[h]
\centering
\begin{tabular}{|c|c|c|c|}
    \hline
         & SS, prior $\tilde{N} = 4$,000  & SS, prior $\tilde{N} = 2$,000 & \chords \\
         \hline \hline 
     $\theta=0$ & 0.22 (0.28)  &  0.28 (0.22) & 0.5 (0.04) \\
     $\theta=0.5$ & 0.31 (0.19) & 0.38 (0.12) & 0.5 (0.06) \\
     $\theta=1$ & 0.43 (0.07) & 0.5 (0.01) & 0.56 (0.11) \\
     Misspecified & 0.33 (0.17) & 0.4 (0.1) & 0.52 (0.07)  \\
    \hline
  \end{tabular}
  \caption{Estimation of the fraction of individuals with degree 10 ($f_{10} = 0.5$). The average estimate is presented in each case, with mean absolute error (MAE) in parentheses.\label{tab:MAE}}
\end{table*}

Table~\ref{tab:MAE} presents $\hat{f}_{10}$, i.e., the estimated frequency of degree $10$, where $f_{10}=0.5$.
This is motivated by estimation of the prevalence of a disease that is mostly present in individuals with high degree.
As can be seen, unless the inverse-degree weighting assumption holds ($\theta=1$), the SS estimator of $f_{10}$ is biased, whereas our estimator is not, even in the misspecified case. 
When $\theta=1$, the \emph{standard} inverse-degree weighting estimator also has similar performance (data not shown).

\begin{figure}[h]
\centerline{
	\includegraphics[width=10cm]{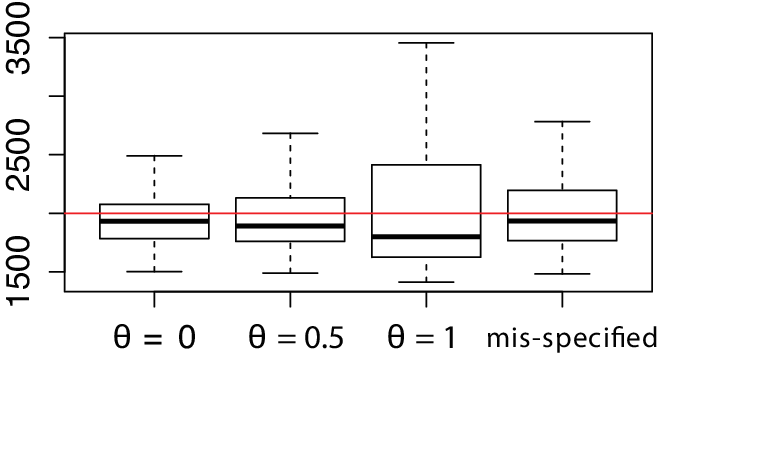}
	}
	\caption{Distributions of estimated population size, using \chords\ with different values of the generative parameter $\theta$. True $N$ equals 2,000; sample size is 1,000; number of replications is 300.
	\label{fig:boxes}
	}
\end{figure}

Figure~\ref{fig:boxes} presents the estimated population size, $\hat{N}$. 
We do not compare it to the SS estimator since, when using SS, $N$ is a required input and cannot be estimated. 
It can be seen that, in the various sampling regimes we tested, we are able to estimate the underlying population size quite accurately.

\subsubsection{Real-world RDS dataset analysis} 
\label{rwRDS} 

In general, RDS is applied to populations where the true size ($N$) and composition is not known, and so real-world datasets that allow us to evaluate our method are scarce. 
We analyzed one such dataset, the WebRDS data, which was made up of a web-based respondent-driven sample of students from a large US university, conducted in 2008 \citep{Wejnert2009}. 
As part of an evaluation of RDS methodology, an RDS survey was carried out on this population, recruiting 378 participants. 
In addition, official institutional records were used (see \citep{Wejnert2009}) to obtain the effective number of students available for recruitment: $N \approx 11$,500. 
Because there is no additional data available for the degree distribution, such as a simple random sample, we used the network from the ``Facebook100'' dataset \citep{Facebook100}, providing a snapshot of the Facebook friendship networks in the university from September 2005. Although this is far from a perfect proxy, we feel the web-like nature of the Facebook friendship network, which is similar to the WebRDS connections, should provide a good approximation to the real degree distribution --- in particular, by having a large mean degree of $84.8$, a standard deviation of $86$, and a heavy tail (Fig.~\ref{fig:USuni}, black line).

There are two difficulties associated with the WebRDS data. 
First, the observed degree distribution is very heterogeneous, and reaches a maximum degree of 1,000, resulting in a small number of recruits per degree class (despite the apparent heaping, where the reported degrees are clustered at 10s, 50s and 100s; Fig.~\ref{fig:USuni}, gray line). 
Second, \citet{Wejnert2009}, who was not interested in the temporal aspects of the process, changed the recruitment scheme midway. As the author acknowledged, this resulted in a change to the recruitment rate, further exacerbated by reaching spring break. Therefore, we first used bins of size $40$ for degree (which was indicated to be optimal by having the best likelihood and AIC score\footnote{Another advantage of our framework is that it lends itself well to a well-established, systematic, and objective means for model selection, such as the AIC. More details on the benefits and possible limitations are given in Section~\ref{sec:future_research}.}) and Jeffrey's prior for regularization \citep{firth1993bias}. 
We then applied our method, and despite the temporal inconsistencies in recruitment, we found a reasonable estimate of $\widehat{N} = 3$,410. 
Moreover, when examining the cumulative degree distribution, the fit with the surrogate degree distribution was surprisingly accurate. 
Figure~\ref{fig:USuni} demonstrates that the inverse-degree-weighted estimator overestimates the fraction of \emph{low} degrees.
Even the observed degree distribution is a better fit than inverse-degree, although it overestimates the fraction of \emph{high} degrees. 
Our new method lies in between the two, adhering to the true distribution.

\begin{figure}[h]
\centerline{
	\includegraphics[width=13cm]{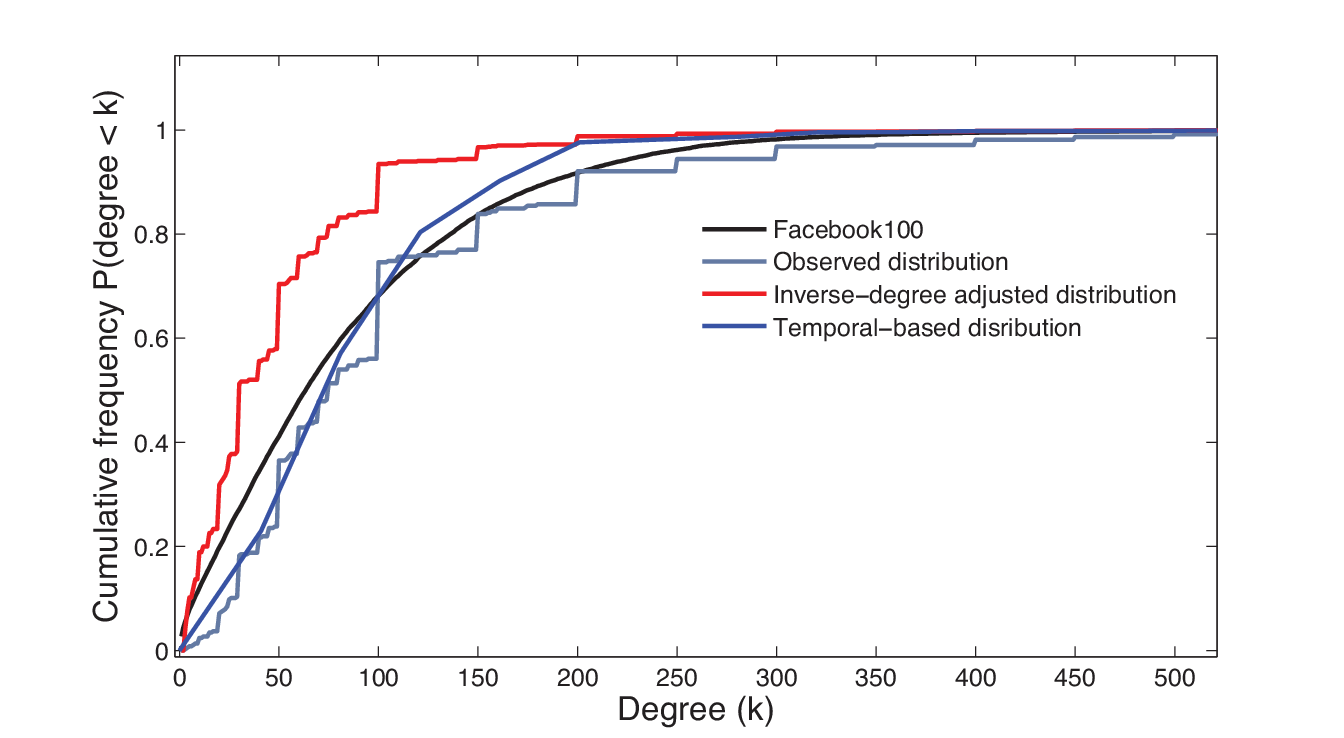}
	}
\caption{Cumulative degree distributions of students. The ``true'' distribution (Facebook100 surrogate) is shown in black. Observed CDF is in gray. The maximum likelihood estimate, calculated with \chords, is in blue. The inverse-degree-weighted estimate is in red. \label{fig:USuni}  }
\end{figure}

\subsection{Asymptotic Analysis}
\label{sec:asymptotics}

A model with an intensity like (M4') usually leads to the following textbook likelihood function~\citep{Andersen93, Aalen1978}:
\begin{equation} 
	\label{like}
	L(\vec{\beta}, \vec{N} ; t) := \exp \Bigg[ \int_0^t \sum_k \log \lambda_{k,t}(\beta_{k}, N_{k})dn_{k,t} -  \int_0^t \sum_k \lambda_{k,t}(\beta_k, N_k)dt\Bigg].
\end{equation}
The large-sample behavior of maximum likelihood estimators in this setup requires some technical preliminaries (\S\ref{asy}).
After these are specified, we provide the necessary details from Kurtz's theory of density-dependent processes, and demonstrate convergence of the counting process to a deterministic function (\S\ref{asycon}). 
These aspects are then put together in Section~\ref{asyT}, for our large-sample theorems. Proofs are provided in the Appendix.

\subsubsection{Technical preliminaries} 
\label{asy}

In our model, the number of individuals, $N$, is one of the parameters. As discussed previously in similar settings (see \citep[p.~430]{Andersen93}, and \citep{VanPul}), it is obviously trivial to conduct a large-sample analysis by considering a sequence of models with $N$ fixed. A more informative large-sample situation is one in which there are many individuals in each degree class.
We therefore consider a sequence of RDS models, indexed by $v$, and by introducing a dummy variable, $f_k$, we let $v f_k$ denote the size of each degree class, $N_k$. In particular, the sequence of counting processes, $n^v_t$, is the multivariate collection of the univariate processes $n^v_{k,t}$, with $v\rightarrow \infty$. 
Similarly, the intensities, Eq.(\ref{intense}), and the likelihood, Eq.(\ref{like}), are indexed mutatis mutandis.
Now we consider the estimation of the $f_k$'s and $\beta_k$'s, as $v\rightarrow\infty$. The result can be later rephrased in terms of the $N_k$'s, with an analogous consistency and asymptotic normality. 
More precisely, denoting convergence in probability by $\limprob$, and convergence in distribution by $\limdist$, 
the consistency of $\hat{f}_k$ (or $\hat{f}_k \limprob f_k$ as $v \rightarrow \infty$) implies $\frac{\hat{N}_k}{N_k} \limprob 1$ as $N_k \rightarrow \infty$, and similarly, concerning asymptotic normality,
\begin{equation}\label{eqasy}
 \sqrt{v}(\hat{f}_k - f_k)\underset{v\rightarrow\infty}{\overset{D} {\longrightarrow}} \mathcal{N}(0,\sigma^2_{k,\vec{\beta},\vec{f}}) \Longrightarrow \frac{1}{\sqrt{N}}({\hat{N}_k}-{N_k})\underset{N, N_k\rightarrow\infty}{\overset{D} {\xrightarrow{\hspace*{1.5cm}}}} \mathcal{N}(0,{\sigma^2_{k,\vec{\beta},\vec{f}}}).
\end{equation}
The formulation in Eq.(\ref{eqasy}) makes it clear that the number of different degree classes in the data cannot grow too fast (in order to avoid having too few observations from each degree class). The simplest and crudest restriction, which we focus on here for simplicity, is one in which the maximal degree, $d_{max}$, is bounded by some constant, $M$, $\forall v$.  This implies the more general and important condition\footnote{We write $f(N)= \Theta(g(N))$ if $\frac{f(N)}{g(N)}\rightarrow const > 0$.} $\forall i,j: \beta_iN_i = \Theta(\beta_jN_j)$, given that $\beta_i$ is fixed for all $i$.

In general, the process $I_t$ can evolve in an arbitrarily complicated manner\footnote{As long as it is adapted to the self-exciting history of $n_t$.}. For example, in the SIR model in epidemiology, each infected individual gets removed at a rate, $\gamma$, which is also of interest. However, since this removal process is both observed fully and uninteresting to us, we will skip modeling it here, and treat the rather general case where $I_t\equiv I_0 +v g(v^{-1}n_{t},t)$, where $g$ is a nonnegative continuous function and $I_0$ is the initial number of seeds used for recruitment (see T1 below).

For simplicity, we tacitly treat the observation period as being $[0,\tau]$, with $\tau$ as a finite number; however, a more general approach is to allow for an observation period $[0, T^v]$, with $T^v$ being stopping times tending to $\tau$ in probability as $v$ increases. In particular, taking $N_{min}:=\min_k\set{N_k}$, we define
\begin{equation} 
\label{stopping}
	\tau := \inf \set{t: \sum_k n_{k,t} =  N_{min}} .
\end{equation}
Although Eq.(\ref{stopping}), implying prior knowledge of $N_{min}$, may appear to be a peculiar stopping time which could be easily weakened, we chose to keep it in order to avoid otherwise necessary distractions from our main point. In particular, this enables us not to specify a particular $I_t$, and preserve the very general condition stated in (T1).


Finally, 
we define the stochastic process $x_v (t)$ as
\begin{equation} \label{xv}
	x_v(t) := v^{-1}n^v_t.
\end{equation}
In many practical situations, as is shown below, this stochastic process converges uniformly on $[0 , \tau]$, in probability, to a deterministic function, $x_{\infty}(t)$, as $v\rightarrow \infty$. 
In order to apply Kurtz's Theorem (Kurtz's law of large numbers) \citep{Kurtz83} and obtain this convergence, it is customary, for example in the study of stochastic epidemics, to have the dynamics (i.e., epidemic) initiated by a positive fraction of the population.
In other words, even though $I_0$ might be a very small fraction of the entire population, we still have $I_0 = \Theta(N)$.

Summarizing all the technical details of this section, we have:
\begin{description}
  \item[(T1)] $v^{-1}I_t\equiv v^{-1}I_0 +g(v^{-1}n_{t},t)$ where $g$ is a non-negative continuous function.
  \item[(T2)] $\frac{I_0}{N}\rightarrow const > 0$.
  \item[(T3)] The maximal degree, $d_{max}$, is bounded by some constant, $M$, $\forall v$.
  \item[(T4)] The observation period $[0, \tau]$ satisfies:
  \begin{equation*} \label{tau}
    \tau := \inf_t \set{\sum_k n_{k,t} = N_{min}}.
    \end{equation*}
\end{description}

\subsubsection{Convergence to a deterministic function} 
\label{asycon}
For purposes of notational convenience, let us write momentarily the parameter space $(\vec{\beta},\vec{N})$ as $\Phi$. Let $K:=D([0, \tau])$ be the Skorokhod space composed of right-continuous functions on $[0, \tau]$ with left limits. The theory developed by Kurtz for the so-called density-dependent process \citep{Kurtz83} deals with processes having an intensity function\footnote{Recall that the superscript $v$ indexes the sequence of processes, each of which evolves in time (subscript $t$) and depends on the parameters $\phi \in \Phi$. The underlying process can be multivariate, and if we need to emphasize one of its components we can go further and write $\lambda^v_{k,t}(\phi)$.}
\begin{equation} 
\label{Xfun}
\lambda^v_t(\phi) = vX(t,\phi,  v^{-1}n_t^v)
\end{equation}
where $\phi \in \Phi$ and $X:=[0,\tau]\times\Phi\times K\mapsto \mathbb{R}^+$ can be a fairly general function depending on the past of the stochastic process, up to, but not including, time $t$. In the multivariate case, Eq.(\ref{Xfun}) means
\begin{equation}
\label{multiXfun}
\lambda^v_t = vX(t,\phi,   v^{-1}n_t^v) := v(X_1(t,\phi,  v^{-1}n_t^v), X_2(t,\phi, v^{-1} n_t^v),\dots, X_{d_{max}}(t,\phi,  v^{-1}n_t^v))
\end{equation}

Using our model of RDS and definitions Eq.(\ref{intense}) and Eq.(\ref{multiXfun}), we now have for the $j^\mathrm{th}$ component of $X$ in RDS
\begin{equation}
\label{Xjfun}
X_j(t,\phi,   v^{-1}n_t^v):=  X_j(t, \vec{\beta}, \vec{N},   v^{-1}n_t^v) =
  \beta_j\frac{I_{t}^-}{v}(f_j - \frac{n_{j,t}^-}{v})
\end{equation}
which is compatible with Eq.(\ref{Xfun}) if, for example, $\frac{I_{t}}{v}$ is a function of $v^{-1}n_t^v$, as guaranteed by (T1--2) in the simplest case.

Two important properties of $X$ as defined in Eqs.(\ref{Xfun}--\ref{Xjfun}) are:
\begin{description}
  \item[(P1)] For all $x\in K$ and for all $\phi\in \Phi$ the function $X$ satisfies:
	\begin{equation*}
	  \sup_{t\leq\tau} X(t,\phi,x)<\infty.
	\end{equation*}
  This, as well as P2 below, hold elementwise, i.e., for all $X_j$.
  \item[(P2)] Lipschitz continuity: there exists a constant, $L$, not depending on $t$, such that for all $x,y\in K$ and all $t\in [0, \tau]$,
 \begin{equation*}
   |X(t,\phi,x)- X(t,\phi,y)| \leq L \sup_{s\leq t} |x(s)-y(s)|.
\end{equation*}
\end{description}
This makes it possible to apply Kurtz's law of large numbers.

\newtheorem{Lemma}{Lemma}
\begin{Lemma}
Let $\phi_0$ be the true value of the parameter $\phi\in\Phi$. The process $x_v(t)$, as defined via Eq.(\ref{xv}), converges uniformly on $[0, \tau]$, in probability, to $x_{\infty}(t)$, as $v\rightarrow\infty$, where $x_{\infty}(t)\in D([0, \tau])$ is the unique solution of
 \begin{equation*}
x(t)= \int_0^t X(s,\phi_0,x)ds.
\end{equation*}
\end{Lemma}
\begin{proof}
This is an immediate result of Kurtz's law of large numbers \citep{Kurtz83}; see, for example, Theorem II.5.4 in \citep{Andersen93}.  \end{proof}
\begin{Remark}
Note that (T1) easily provides similar convergence of $v^{-1}I^v_t$ to some deterministic function, $I^{\infty}(t)$.
\end{Remark}

The following properties of $X$ can now also be shown to hold:
\begin{description}
  \item[(P3)] There exist neighborhoods $\Phi_0$ and $K_0$, of $\phi_0$ and $x_{\infty}$ respectively, such that the function  $X(t,\phi,x)$ and its derivatives with respect to $\phi$ of the first, second and third order exist, are continuous functions of $\phi$ and $x$, and are bounded on $[0,\tau]\times \Phi_0\times K_0$.
  \item[(P4)] The function  $X(t,\phi,x)$ is bounded away from zero on $[0,\tau]\times \Phi_0\times K_0$.
  \item[(P5)] For $1\leq i \leq 2d_{max}$, let $\phi^i$ denote the parameter $f_i$ if $1\leq i\leq d_{max}$, and otherwise denote the parameter $\beta_{i-d_{max}}$ if $d_{max}+1\leq i\leq 2d_{max}$.   The matrix $\Sigma=\{\sigma_{ij}(\phi_0)\}$ is positive definite, with, for $1\leq i,j \leq 2d_{max}$ and $\phi\in \Phi_0$,
  \begin{equation*}
\sigma_{ij}(\phi):= \int_0^{\tau} \sum_k    \frac{ \frac{\partial}{\partial \phi^i}X_k(s,\phi,x_{\infty}) \frac{\partial}{\partial \phi^j}X_k(s,\phi,x_{\infty})}     {X_k(s,\phi,x_{\infty})}  ds.
 \end{equation*}
\end{description}
(P3) is trivial and (P4) is an immediate result of (T4), whereas (P5), a long and straightforward calculation, is dealt with in the Appendix.

Finally, we have everything at hand to present and prove our main results.

\subsubsection{Asymptotic consistency and normality}
\label{asyT}

Summarizing previous sections, it is not possible to derive any informative large sample approximation by considering a sequence of models with $N$ fixed.
We therefore consider a sequence of RDS models, indexed\ by $v$. Introducing a dummy variable, $f_k$, we let $v f_k$ denote the size of each degree class, $N_k$ (see in \S\ref{asy} the discussion leading to Eq.(\ref{eqasy})).
The likelihood function, Eq.(\ref{like}), is thus also indexed by $v$, and rewritten as
\begin{equation} 
\label{likeV}
L_v(\vec{\beta}, \vec{N} ; t) := \exp \Bigg[ \int_0^t \sum_k \log v X_{k}(t, \phi, x_v)dn^v_{k,t} - v \int_0^t \sum_k X_{k}(t, \phi, x_v)dt\Bigg]
\end{equation}
with $X_{k}$ as defined in Eq.(\ref{multiXfun}) and Eq.(\ref{Xjfun}), and $x_v$ as in Eq.(\ref{xv}). Similarly, we define the log-likelihood function
\begin{equation} \label{logL}
C_v(\vec{\beta}, \vec{N} ; t) := \log L_v
\end{equation}
and the negative observed information matrix,
\begin{equation} \label{Iij}
\mathcal{I}^{i,j}_v(\vec{\beta}, \vec{N} ; t) := \frac{\partial^2}{\partial \phi^i\partial \phi^j}C_v
\end{equation}
with $\phi^i$ as defined in (P5).

Our main theorems are as follows.

\newtheorem{Theorem}{Theorem}
\begin{Theorem}[Consistency and Normality of Degree Frequencies]
\label{theorem}
Consider a sequence of RDS counting processes (M1--3) with intensity function (M4') and $( \vec{N}, \vec{\beta} )$ as parameters. We index the sequence with $v\rightarrow \infty$, obtaining a reparameterization $(\vec{f}, \vec{\beta} )$, with  $(\vec{f_0}, \vec{\beta_0} )$ as the true (unknown) values. If conditions (T1--4) hold then there exists a unique consistent solution $(\vec{\hat{f_v}}, \vec{\hat{\beta_v}} )$, to the score equations, $\frac{\partial}{\partial \phi^i}C_v(\phi,\tau)=0$. Moreover, this solution provides a local maximum of the likelihood function, Eq.(\ref{likeV}), and

\begin{equation*} \label{bjhkbjh}
\sqrt{v}\bigg((\vec{\hat{f_v}}, \vec{\hat{\beta_v}} )-(\vec{f_0}, \vec{\beta_0} )\bigg)^T \underset{v\rightarrow\infty}{\overset{D} {\longrightarrow}} \mathcal{N}(0,\Sigma^{-1})
\end{equation*}
where $\Sigma$, given by (P5), can be estimated consistently from the observed information matrix, Eq.(\ref{Iij}).
\end{Theorem}

\begin{proof}
Since  the intensity, $\lambda^v_t(\phi)$, can be written as $ vX(t,\phi,  v^{-1}n_t^v)$ where $X$ fulfills conditions (P1--5), this is an immediate result of  Theorems VI.1.1 and VI.1.2 of \citet{Andersen93}. See, for example, the less general Theorem 1 of \citet{VanPul}, which is similar to our case (P1--4 were treated in \S\ref{asycon}, and P5 is shown in the Appendix).
\end{proof}

As discussed earlier, researchers working on RDS are typically not interested in the degree distribution per se, but rather in the prevalence of, for example, HIV.
Having obtained an estimate, $\widehat{f}$, of the degree distribution, it is straightforward to stratify and weight the observations to obtain an estimate,
\begin{equation} \label{Hestimate}
\widehat{H} = \sum_{i=1}^n \widehat{f}_{d_i}\frac{Y_i}{n_{d_i}}
\end{equation}
where $n$ is the sample size, $d_i$ the degree of the $i^\mathrm{th}$ individual in the sample, $n_{k}$ is the number of individuals in the sample having degree $k$, $Y_i=1$ if individual $i$ is HIV-infected, and  $Y_i=0$ otherwise. 
Eq.(\ref{Hestimate}), which takes an individual-based viewpoint, could also be calculated with a degree-class perspective in mind:
\begin{equation} \label{HDestimate}
\widehat{H} = \sum_{i=1}^k \widehat{f_k}\frac{n^+_k}{n_{k,\tau}}
\end{equation}
where $n^+_k$ is the number of HIV infected individuals in the sample having degree $k$.

We denote by $N^+_k$ the number of HIV-infected individuals in the population having degree $k$.
It might be safe to assume that the distribution of $\sqrt{v}(\frac{n^+_k}{n_{k,\tau}}-\frac{N^+_k}{N_k})$ is well approximated by a normal distribution  $\mathcal{N}(0,\sigma^{2}_{\hat{p}_k})$ independent of everything else, where subscript $p_k$ serves to indicate that this is the variance of $\hat{p}_k:=\frac{n^+_k}{n_{k,\tau}}$, the estimator of  $p_k:= \frac{N^+_k}{N_k}$, the \emph{prevalence} within degree class $k$.

\begin{Remark}
For example, if $n^+_k$ is distributed hypergeometrically, viz.\ $HG(N_k,N^+_k,n_{k,\tau};n^+_k)$, we can even justify $\sigma^2_{\hat{p}_k} = \frac{1}{n_{k,\tau}}\frac{N^+_k}{N_k} \frac{N_k - N^+_k}{N_k} \frac{N_k-n_{k,\tau}}{N_k-1} $.
\end{Remark}

Similarly, we use $\sigma^{2}_{\hat{f}_k}$ to denote the variance of the estimator of $f_k$ (see (P5) and the proof of Theorem~\ref{theorem}). Our second main result is:

\begin{Theorem}[Consistency and Normality of Prevalence]
\label{corollary}
$\widehat{H}$, the  prevalence estimator given by Eq.(\ref{Hestimate}), is asymptotically consistent and normally distributed:
\begin{equation*} 
\label{hhh}
	\sqrt{v}(\widehat{H} - H_0) \underset{v\rightarrow\infty}{\overset{D} {\longrightarrow}} \mathcal{N}\left(0,\sum_k p_k^2 \sigma^{2}_{\hat{f}_k}   + \sum_k f_k^2 \sigma^{2}_{\hat{p}_k} \right)
\end{equation*}
where $H_0$ is the true prevalence within the population.
\end{Theorem}
\begin{proof}
A simple application of the delta method (see Appendix).
\end{proof}

\section{Discussion}
\label{sec:discussion}

In this work, we have introduced a model-based approach for estimating degree frequencies and population size, using the recruiting time in RDS. We have demonstrated its empirical performance, and proven its large-sample properties.

The new approach presented here, and its underlying assumptions, are more parsimonious and easier to control, verify, and correct than the implicit assumptions underlying the current state-of-the-art, inverse-degree weighting approach. 
By incorporating additional data into a model-based approach, we do not merely replace an old set of assumptions with new ones. 
Our underlying assumptions can be tested and improved (applying, e.g., AIC or FIC), while also addressing the identifiability concerns in inverse-degree weighting.

Here, we addressed only the simplest case. 
More elaborate models could be constructed that account for homophily, for example --- using recruitment probabilities that depend on the state of both recruiter and recruit --- and covariates other than degree. 
Moreover, Bayesian regularization is possible as well: often some prior on the population size, or the degree distribution, is available.
Using a prior for the degree distribution in a fairly straightforward manner might alleviate the difficulties of moderate-sized samples. 
For example, our WebRDS estimates (\S\ref{rwRDS}) are regularized using Jeffrey's prior.

Finally, another advantage of our approach is its ease of use, and integration with current methods and protocols for sampling in RDS. 
There is practically no need for design adjustments, apart from careful recording of interview times. 
We assumed a simple recruitment process for ease of exposition.  
Inserting new seeds during sampling, limiting the number of coupons, and other adjustments, can be easily accommodated in our framework (see \S\ref{asy} addressing (T1)).

\subsection{Future Research}
\label{sec:future_research}
Returning to the analogous mathematical models of infectious disease transmission, it is clear that simple initial epidemiological models, introduced over a century ago, do not capture the transmission dynamics of all diseases in general.
Over the years, many extensions to, and variations on, the first models have arisen. Similarly, it is unlikely that the recruitment process in RDS surveys is governed merely by the degree (and further, that the relationship is \emph{proportional}). 
Incorporation of covariates often collected during RDS, such as age, race/ethnicity, socioeconomic status, and coupon value, is an avenue for future research. 
An adaptation of our approach to an alternative single covariate is straightforward with \chords, simply by replacing the notion of degree classes with the new kind of class: for example, age classes such as children/adults/seniors. 
The models can also be extended to account for homophily, by having a recruitment rate for each class of recruiter--recruit relation. 
By employing a likelihood-based model, there are many approaches that could be used, in either a frequentist or Bayesian setting, to test important determinants of the recruitment rate that may also impact estimates of interest. 
In addition to establishing new theory, the simple numerical optimization that the accompanying software package currently supports might not suffice for more elaborate models, and would require updating.

\section{Appendix}
Theorems~\ref{theorem} and \ref{corollary} are derived in detail below.

For Theorem~\ref{theorem}, we first need to show that (P5) holds.
The matrix $\Sigma$ contains the four blocks
$\bigl(\begin{smallmatrix}
A&B\\ C&D
\end{smallmatrix} \bigr)$
with $A$, for example, depicting the association between the different $f_i$'s, and $D$ depicting the association between the different $\beta_i$'s.
A simple calculation shows that the matrix $A$ is a diagonal matrix with entries
\begin{equation} \label{Aii}
 a_{ii} = \int_0^\tau  \frac{\beta_i I_s^{\infty}}{f_i - n^{\infty}_{i,s}}ds
\end{equation}
where $I_s^{\infty}$ is the deterministic function that $v^{-1}I_s^{v}$ converges to, and similarly $n^{\infty}_{k,s}$ is the deterministic function that $v^{-1}n^{v}_{k,s}$ converges to.  Note we have omitted the ``just before'' notation, $g^-$, from $I$ and $n$. This may be done since the integration is with respect to Lebesgue's measure, $ds$.
Similarly, the matrices $B,C,D$ are also diagonal matrices with entries
\begin{equation} 
\label{Bii}
 b_{ii}=  c_{ii} = \int_0^\tau  I_s^{\infty}ds,
\end{equation}
\begin{equation} 
\label{Dii}
 d_{ii} = \int_0^\tau  \frac{I_s^{\infty}(f_i - n^{\infty}_{i,s})}{\beta_i}ds.
\end{equation}
The invertibility of $\Sigma$ can be demonstrated by a direct calculation of its inverse, using the fact that
\begin{equation} \label{inverse}
\left(
  \begin{array}{cc}
    A & B \\
    C & D \\
  \end{array}
\right)^{-1} = \left(
  \begin{array}{cc}
    (AD - BC)^{-1} &0\\ 0 & (AD - BC)^{-1}
  \end{array}
\right)
\left(
  \begin{array}{cc}
    D&-B\\ -C&A
  \end{array}
\right)
\end{equation}
Thus we need only show that $AD - BC$ is invertible, i.e., that for all $i$,
\begin{equation} \label{Aii-Dii}
\int_0^\tau  \frac{\beta_i I_s^{\infty}}{f_i - n^{\infty}_{i,s}}ds  \int_0^\tau  \frac{I_s^{\infty}(f_i - n^{\infty}_{i,s})}{\beta_i}ds - \int_0^\tau  I_s^{\infty}ds \int_0^\tau  I_s^{\infty}ds \neq 0
\end{equation}
But for the first term in Eq.(\ref{Aii-Dii}), we have
\begin{equation*} \label{cccc}
\int_0^\tau  \frac{\beta_i I_s^{\infty}}{f_i - n^{\infty}_{i,s}}ds  \int_0^\tau  \frac{I_s^{\infty}(f_i - n^{\infty}_{i,s})}{\beta_i}ds
=
\int_0^\tau  \Bigg(\sqrt{\frac{ I_s^{\infty}}{f_i - n^{\infty}_{i,s}}}\Bigg)^2 ds  \int_0^\tau  \bigg(\sqrt{{I_s^{\infty}(f_i - n^{\infty}_{i,s})}}\bigg)^2ds
 >
 \Bigg(\int_0^\tau  I_s^{\infty}ds\Bigg)^2
\end{equation*}
with strict inequality after applying the Cauchy--Schwarz inequality  for two nonlinearly dependent functions.

This allows us to apply Theorems VI.1.1 and VI.1.2 of \citet{Andersen93}, and establish the existence part. The following calculations demonstrate the uniqueness of the solution.

Differentiating the log-likelihood gives
\begin{equation} \label{dbk}
\frac{dC}{\partial\beta_k} = \frac{n_{k,\tau}}{\beta_k} - \int_0^\tau \frac{1}{N} (N_k - n_{k,t}^-)I_t^- dt,
\end{equation}
and
\begin{equation} \label{dNk}
\frac{dC}{\partial N_k} = \int_0^\tau \frac{1}{N_k - n_{k,t}^-}dn_{k,t} - \int_0^\tau \frac{1}{N} \beta_k I_t^-  dt.
\end{equation}
These can be written as pairs of equations with two unknowns, and equated to $0$, which yields $\hat{N}_k$, a MLE for $N_k$, via the (unique) solution to
\begin{equation}\label{Nkeq}
 \sum_{i=0}^{n_{k,\tau}-1} \frac{1}{\hat{N}_k - i} = \frac{n_{k,\tau} \int_0^\tau I_t dt }{\hat{N}_k \int_0^\tau I_t dt - \int_0^\tau n_{k,t} I_t dt}.
\end{equation}
Rearranging Eq.(\ref{Nkeq}) we get
\begin{equation}
 \sum_{j=0}^{n_{k,\tau}-1}
\prod_{\substack{0\leq i<n_{k,\tau} \\ i\neq j }} ( \hat{N}_k - i)\Bigg( \hat{N}_k \int_0^\tau I_t dt - \int_0^\tau n_{k,t} I_t dt\Bigg) - n_{k,\tau} \int_0^\tau I_t dt \prod_{j=0}^{n_{k,\tau}-1} ( \hat{N}_k - j) = 0,
\end{equation}
and
\begin{equation}
\Bigg( \hat{N}_k \int_0^\tau I_t dt - \int_0^\tau n_{k,t} I_t dt\Bigg)\sum_{j=0}^{n_{k,\tau}-1} \frac{\hat{N}_k!}{(\hat{N}_k - n_{k,\tau})!(\hat{N}_k - j)} -  \frac{n_{k,\tau} \int_0^\tau I_t dt \hat{N}_k!}{(\hat{N}_k - n_{k,\tau})!} = 0.
\end{equation}
Dividing by $\frac{\hat{N}_k!}{(\hat{N}_k - n_{k,\tau})!}$ and substituting $\sum_{j=0}^{n_{k,\tau}-1} 1$ for $n_{k,\tau}$, we get
\begin{equation}
\Bigg( \hat{N}_k \int_0^\tau I_t dt - \int_0^\tau n_{k,t} I_t dt\Bigg)\sum_{j=0}^{n_{k,\tau}-1} \frac{1}{(\hat{N}_k - j)} -  \sum_{j=0}^{n_{k,\tau}-1}  \int_0^\tau I_t dt = 0,
\end{equation}
or
\begin{equation}
\mathlarger{\mathlarger{‎‎\sum}}_{j=0}^{n_{k,\tau}-1} \frac{\bigg( \hat{N}_k \int_0^\tau I_t dt - \int_0^\tau n_{k,t} I_t dt\bigg)}{(\hat{N}_k - j)} -  \sum_{j=0}^{n_{k,\tau}-1}\frac{(\hat{N}_k - j) \int_0^\tau I_t dt}{(\hat{N}_k - j)} = 0,
\end{equation}
and finally
\begin{equation}\label{final}
\mathlarger{\mathlarger{‎‎\sum}}_{j=0}^{n_{k,\tau}-1} \frac{\bigg( j \int_0^\tau I_t dt - \int_0^\tau n_{k,t} I_t dt\bigg)}{(\hat{N}_k - j)} = 0.
\end{equation}
Let $N_k^\star$ be a solution to Eq.(\ref{final}), and assume in contradiction that Eq.(\ref{final}) has another solution, $N_k^\bullet < N_k^\star$, in the range $\hat{N}_k\geq n_{k,\tau}$. Notice that the tail of the sum in Eq.(\ref{final}) is comprised of positive terms, and perhaps an irrelevant zero term, which are required to cancel out the negative terms comprising the beginning of the sum. We write this as $Head(\hat{N}_k)+Tail(\hat{N}_k) = 0$, emphasizing the functional relationship of Eq.(\ref{final}) and $\hat{N}_k$. Let $j^\star$ be the index of the first positive term in the sum in Eq.(\ref{final}); in other words, $Head(\hat{N}_k)$ is comprised of $j^\star$ negative terms (ignoring the possibility of an irrelevant zero term).

Since $N_k^\star$ is a solution to Eq.(\ref{final}), we have $Head(N_k^\star)+Tail(N_k^\star) = 0$. However, $Tail(N_k^\bullet)\geq \frac{N_k^\star-j^\star}{N_k^\bullet-j^\star} Tail(N_k^\star)$, while  $| Head(N_k^\bullet)|\leq \frac{N_k^\star-j^\star+1}{N_k^\bullet-j^\star+1}|Head(N_k^\star)|$. We finish by concluding  $Head(N_k^\bullet)+Tail(N_k^\bullet) \neq0$, i.e., a contradiction.

For Theorem~\ref{corollary}, we apply the delta method, assuming that $\sqrt{v}(\hat{p}_k - p_k)$ converges to  $\mathcal{N}(0,\sigma^{2}_{\hat{p}_k})$, and since from Theorem~\ref{theorem} we have
$\sqrt{v}(\hat{f}_k - f_k)  \underset{v\rightarrow\infty}{\overset{D} {\longrightarrow}} \mathcal{N}(0,\sigma^{2}_{\hat{f}_k})$. See also Eqs.(\ref{Aii}--\ref{inverse}).

We define the concatenated random vector $\hat{\phi}:=(\overrightarrow{\hat{f},\hat{p}})^T=(\hat{f}_1,\hat{f}_2,\dots,\hat{f}_{d_{max}},\hat{p}_1,\hat{p}_2,\dots,\hat{p}_{d_{max}})^T$, with a covariance matrix $C:=diag\{ \sigma_{\hat{f}_1}^2,\sigma_{\hat{f}_2}^2,\dots,\sigma^2_{\hat{f}_{d_{max}}},\sigma^2_{\hat{p}_1},\sigma^2_{\hat{p}_2},\dots,\sigma^2_{\hat{p}_{d_{max}}}  \}  $. Rewrite $\widehat{H}$ as the function
\begin{equation*} \label{hfun}
h(\overrightarrow{\hat{f},\hat{p}})=\sum_k \hat{f}_k \hat{p}_k
\end{equation*}
with $\nabla h(\overrightarrow{\hat{f},\hat{p}}) = (\hat{p}_1,\hat{p}_2,\dots,\hat{p}_{d_{max}},\hat{f}_1,\hat{f}_2,\dots,\hat{f}_{d_{max}})^T$. From the delta method, we have
\begin{equation*} \label{hbhhj}
\sqrt{v}(h(\hat{\phi}) - h(\phi_0))
 \underset{v\rightarrow\infty}{\overset{D} {\longrightarrow}} \mathcal{N}(0, \nabla h(\phi_0)^T C \nabla h(\phi_0) )
 \end{equation*}
 where $\phi_0$ contains the true values of the associated estimated parameters. Since $h(\phi_0)=H_0$ and $\nabla h(\phi_0)^T C \nabla h(\phi_0)=\sum_k p_k^2 \sigma^{2}_{\hat{f}_k}   + \sum_k f_k^2 \sigma^{2}_{\hat{p}_k}$, Theorem~\ref{corollary} is now established. \hfill $\square$

\newpage
\bibliographystyle{plainnat}
\bibliography{Ucrc}

\end{document}